\documentclass{article}

\usepackage{amsmath,amsthm,mathtools}
\usepackage{graphicx}
\usepackage{algorithm2e}
\usepackage{thm-restate}
\usepackage{enumitem}

\usepackage{tikz}
\usetikzlibrary{shapes, calc, arrows, through, intersections, decorations.pathreplacing, patterns}

\newtheorem{theorem}{Theorem}

\newtheorem{definition}[theorem]{Definition}

\newtheorem{corollary}[theorem]{Corollary}

\newcommand{\mc}{\mathcal}

\let\oldnl\nl
\newcommand{\nonl}{\renewcommand{\nl}{\let\nl\oldnl}}

\title{Three-dimensional matching is NP-Hard}

\author{\normalsize{Shrinu Kushagra}\footnote{Email: skushagr@uwaterloo.ca. The results in this note first appeared in \cite{kushagra19semisupervised} where the reduction was used to prove query lower bounds for a clustering problem.}}
\date{}
\begin{document}
\maketitle
\begin{abstract}
The standard proof of NP-Hardness of 3DM provides a power-$4$ reduction of 3SAT to 3DM. In this note, we provide a linear-time reduction. Under the exponential time hypothesis, this reduction improves the runtime lower bound from $2^{o(\sqrt[4]{m})}$ (under the standard reduction) to $2^{o(m)}$. \\


\end{abstract}

\section{Introduction}
In this note, we first establish the hardness of the following decision problem. 

\vspace{10pt}\begin{definition}[3DM]\textcolor{white}{.}

\vspace{5pt}\noindent Input: Sets $W, X$ and $Y$ and a set of matches $M \subseteq W \times X \times Y$ of size $m$.

\vspace{5pt}\noindent Output: YES if there exists $M' \subseteq M$ such that each element of $W, X, Y$ appears exactly once in $M'$. NO otherwise. 
\end{definition}

\vspace{10pt}\noindent To prove that 3DM is NP-Hard, we reduce an instance of 3SAT to the given problem. Next, we define the 3SAT decision problem. 

\vspace{10pt}\begin{definition}[3-SAT]\textcolor{white}{.}

\vspace{5pt}\noindent Input: A boolean formulae $\phi$ in 3CNF form with $n$ literals and $m$ clauses.

\vspace{5pt}\noindent Output: YES if $\phi$ is satisfiable, NO otherwise. 
\end{definition}

\vspace{10pt}\noindent Given an instance of 3SAT with $n$ literals and $m$ clauses, \cite{garey1979computers} construct a graph with {$\Theta(nm)$ vertices} and $\Theta(n^2 m^2)$ edges. Thus, this is a power-$4$ reduction. In this note, we use a similar but a more efficient gadget and provide a linear time reduction of the 3SAT instance to the given problem. 

\section{Hardness of 3DM}
\label{section:3dmhardness}
\begin{theorem}
Three-dimensional matching is an NP-Hard problem. 
\end{theorem}

\begin{figure}[!ht]
	\centering
	\includegraphics[trim = 120 500 250 120, clip, width=\linewidth]{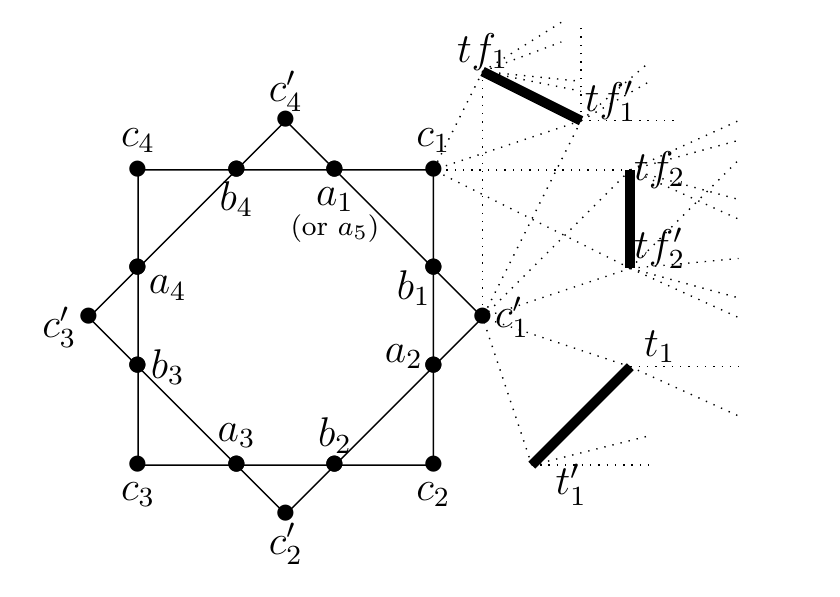}
  	\caption{Part of graph $G$ constructed for the literal $x_1$. The figure is an illustration for when $x_1$ is part of four different clauses. The triangles (or hyper-edge) $(a_i, b_i, c_i)$ capture the case when $x_1$ is true and the other triangle $(b_i, c_i', a_{i+1})$ captures the case when $x_1$ is false. Assuming that a clause $C_j = \{x_1, x_2, x_3\}$, the hyper-edges containing $tf_i, tf_i'$ and $t_1, t_1'$ capture different settings. The hyper-edges containing $t_1, t_1'$ ensure that atleast one of the literals in the clause is true. The other two ensure that two variables can take either true or false values.}
\label{fig:3DMQueries}
\end{figure}

Our reduction is described in Fig. \ref{fig:3DMQueries}. For each literal $x_i$, let $m_i$ be the number of clauses in which the the literal is present. We construct a ``truth-setting" component containing $2m_i$ hyper-edges (or triangles). We add the following hyper-edges to $M$.
\begin{align*}
  &\{(a_k[i], b_k[i], c_k[i]): 1 \le k \le m_i\}\\
  &\cup \{(a_{k+1}[i], b_k[i], c_k'[i]): 1 \le k \le m_i\}
\end{align*}
Note that one of $(a_k, b_k, c_k)$ or $(a_{k+1}, b_k, c_k')$ have to be selected in a matching $M'$. If the former is selected, that corresponds to the variable $x_i$ being assigned true, the latter corresponds to false. This part is the same as the standard construction. 

For every clause $C_j = \{x_1, x_2, x_3\}$ we add three types of hyper-edges.  The first type ensures that atleast one of the literals is true. 
$$\{(c_k[i], t_1[j], t_1'[j]): x_i' \in C_j\} \cup \{(c_k'[i], t_1[j], t_1'[j]): x_i \in C_j\}$$ 
The other two types of hyper-edges (conected to the $tf_i$'s) say that two of the literals can be either true or false. Hence, we connect them to both $c_k$ and $c_k'$
\begin{align*}
  &\{(c_k[i], tf_1[j], tf_1'[j]): x_i' \text{ or }x_i\in C_j\}\\
  &\cup \{(c_k[i], tf_2[j], tf_2'[j]): x_i \text{ or }x_i' \in C_j\}\\
  &\cup \{(c_k'[i], tf_1[j], tf_1'[j]): x_i' \text{ or }x_i\in C_j\}\\
  &\cup \{(c_k'[i], tf_2[j], tf_2'[j]): x_i \text{ or }x_i' \in C_j\}
\end{align*}
Note that in the construction $k$ refers to the index of the clause $C_j$ in the truth-setting component corresponding to the literal $x_i$. Using the above construction, we get that
\begin{align*}
  & W = \{c_k[i], c_k'[i]\}\\
  & X = \{a_k[i]\} \cup \{t_1[j], tf_1[j], tf_2[j]\}\\
  & Y = \{b_k[i]\} \cup \{t_1'[j], tf_1'[j], tf_2'[j]\}
\end{align*} 
Hence, we see that $|W| = 2\sum_i m_i = 6m$. Now, $|X| = |Y| = \sum_i m_i + 3m = 6m$. And, we have that $|M| = 2\sum_i m_i + 15m = 21m$. Thus, we see that this construction is linear in the number of clauses. 

Now, if the 3-SAT formula $\phi$ is satisfiable then there exists a matching $M'$ for the 3DM problem. If a variable $x_i = T$ in the assignment then add $(c_k[i], a_k[i], b_k[i])$ to $M'$ else add $(c_k'[i], a_{k+1}[i], b_k[i])$. For every clause $C_j$, let $x_i$ (or $x_i'$) be the variable which is set to true in that clause. Add $(c_k'[i], t_1[j], t_1'[j])$  (or $(c_k[i], t_1[j], t_1'[j])$) to $M'$. For the remaining two clauses, add the hyper-edges containing $tf_1[j]$ and $tf_2[j]$ depending upon their assignments. Clearly, $M'$ is a matching. 

Now, the proof for the other direction is similar. If there exists a matching, then one of $(a_k, b_k, c_k)$ or $(a_{k+1}, b_k, c_k')$ have to be selected in a matching $M'$. This defines a truth assignment of the variables. Now, the construction of the clause hyper-edges ensures that every clause is satisfiable.

\section{Exponential Time Hypothesis for 3DM}
Before we start the discussion in the section, lets review the definition of the exponential time hypothesis.\\

\vspace{0pt}\noindent\textbf{Exponential Time Hypothesis (ETH)}\\
There does not exist an algorithm which decides 3-SAT  and runs in $2^{o(m)}$ time.\\

\vspace{0pt}\noindent If the exponential hypothesis is true, the standard reduction of 3-SAT to 3DM \cite{garey1979computers} implies that any algorithm for 3DM runs in $2^{o(m^{1/4})}$. However, using the reduction in Section \ref{section:3dmhardness}, we get a more tighter dependence on $m$ stated as a theorem below. 

\begin{theorem}
If the exponential time hypothesis holds then there does not exist an algorithm which decides the three-dimensional matching problem (3DM) and runs in time $2^{o(m)}$.
\end{theorem}
\begin{proof}
For the sake on contradiction, suppose that such an algorithm $\mc A$ exists. Then, using the reduction from Section \ref{section:3dmhardness} and $\mc A$, we get an algorithm for 3SAT that runs in $2^{o(m)}$ time which contradicts the ET hypothesis.  
\end{proof}

\vspace{10pt}\noindent An immediate corollary of this result applies to another popular problem Exact Cover by 3-sets.
\vspace{0pt}\begin{definition}[X3C]\textcolor{white}{.}

\vspace{5pt}\noindent Input: $U = \{x_1, \ldots, x_{3q}\}$. A collections of subsets $S = \{S_1, \ldots, S_m\}$ such that each $S_i \subset U$ and contains exactly three elements.

\vspace{5pt}\noindent Output: YES if there exist $S' \subseteq S$ such that each element of $U$ occurs exactly once in $S'$, NO otherwise. 
\end{definition}

\vspace{5pt}\begin{corollary}
\label{cor:X3CLowerBound}
If the exponential time hypothesis holds then there does not exist an algorithm which decides exact cover by 3-sets problem (X3C) and runs in time $2^{o(m)}$.
\end{corollary}
\begin{proof}
The proof follows from the trivial reduction of 3DM to X3C where $U = W \cup X \cup Y$ and $S = M$. 
\end{proof}

\vspace{20pt}
\bibliography{3dm}
\bibliographystyle{apalike}
\end{document}